\documentclass[conference]{IEEEtran}

\IEEEoverridecommandlockouts

\usepackage{cite}
\usepackage{amsmath,amssymb,amsfonts, amsthm}
\usepackage{algorithmic}
\usepackage{graphicx}
\usepackage{textcomp}
\usepackage{xcolor}

\usepackage{bm}

\usepackage{url}
\usepackage{cite}
\usepackage{float}
\usepackage{caption}
\usepackage{comment}
\usepackage[ruled,linesnumbered]{algorithm2e}

\usepackage{float}
\usepackage{multicol}
\newcounter{multifig}
\usepackage{thmtools,thm-restate}
\usepackage{multirow}
\usepackage{mathtools, cuted}

\usepackage{lipsum}

\usepackage{array}
\usepackage{makecell}

\usepackage[caption=false, font=footnotesize]{subfig}

\def\BibTeX{{\rm B\kern-.05em{\sc i\kern-.025em b}\kern-.08em
    T\kern-.1667em\lower.7ex\hbox{E}\kern-.125emX}}

\newcommand\norm[1]{\left\lVert #1 \right\rVert}

\newtheorem{theorem}{Theorem}[section]
\newtheorem{lemma}[theorem]{Lemma}

\newtheorem{definition}[theorem]{Definition}

\def\BibTeX{{\rm B\kern-.05em{\sc i\kern-.025em b}\kern-.08em
    T\kern-.1667em\lower.7ex\hbox{E}\kern-.125emX}}

\begin{document}

\title{Near-optimal Differentially Private Client Selection in Federated Settings}

\author{\IEEEauthorblockN{Syed Eqbal Alam,
Dhirendra Shukla, and Shrisha Rao}\thanks{Syed Eqbal Alam and Dhirendra Shukla are with the Faculty of Engineering, University of New Brunswick, Fredericton, New Brunswick, Canada. Shrisha Rao is with the International Institute of Information Technology,
Bangalore, Karnataka, India.}}

\maketitle
\begin{abstract}
We develop an iterative differentially private algorithm for client selection in federated settings. We consider a federated network wherein clients coordinate with a central server to complete a task; however, the clients decide whether to participate or not at a time step based on their preferences---local computation and probabilistic intent. The algorithm does not require client-to-client information exchange. The developed algorithm provides near-optimal values to the clients over long-term average participation with a certain differential privacy guarantee. Finally, we present the experimental results to check the algorithm's efficacy.
\end{abstract}



\IEEEkeywordsname{: Differential privacy, Federated optimization, Client selection, Distributed optimization,  Optimization and control.}


         




\footnote{ To appear in the proceedings of the 59th Annual Allerton Conference on Communication, Control, and Computing, September 2023, Monticello, Illinois, USA.}



\section{Introduction}
Let us consider a federated network wherein clients such as mobile phones, IoT devices, etcetera coordinate with a central server or edge servers to complete a task and achieve the social welfare of the network. However, they may not wish to exchange information with other clients in the network. These devices are constrained by battery life,  computational power, network bandwidth, etc. To complete their tasks, they may need additional shared resources. Such settings wherein clients coordinate with a central server and do not require inter-device communication are called {\em federated settings}. Federated setting has recently attracted much interest from the machine learning community as in federated learning \cite{Mcmahan2017, KonecnyMR2015}. 
In federated learning, several client selection strategies are developed; in most approaches, a subset of clients are selected randomly by the central server at a time step to train the global model. The clients download the global model, train it on their device data, and send the learned parameters to the central server without communicating it to other clients in the network. The central server then aggregates the clients' parameters and updates the global model based on the weighted average of the clients' parameters \cite{Kairouz2021}. In these strategies, the central server randomly chooses the clients and does not consider the client's preferences and choices of whether the clients want to participate or not.

We develop a local differential privacy algorithm for client selection in a federated setting in which clients decide whether to participate in completing a task based on their local computation and their probabilistic intent. Following the solution, the clients reach near-optimal solutions over long-term average participation with a differential privacy guarantee. 

Differential privacy (DP) was proposed by Dwork et al. \cite{Dwork2006}, \cite{Dwork2014}. It provides a certain amount of privacy guarantee to clients storing their data in a (centralized) database. Moreover, differential privacy provides a certain amount of resistance to a client's re-identification while interacting with the database; this factor is called {\em privacy budget}. The privacy budget is the maximum amount of information that can be learned about a client from the privacy mechanism's output. 

There are several real-world applications of differential privacy; for example, in smart metering \cite{Sankar2013}, medical imaging \cite{Kaissis2020}, collecting usage statistics on web browser \cite{Erlingsson2014} (deployed by Google), collecting telemetry data of user devices \cite{Ding2017} (deployed by Microsoft), to learn new words from the users' devices \cite{Thakurta2017} (deployed by Apple).

In this paper, we propose a differentially private algorithm for client selection problems in federated settings with a certain privacy guarantee to clients. We extend the work of  \cite{Griggs2016} and introduce differential privacy guarantees to clients in the network. The algorithm provides near-optimal values to the clients over long-term average participation with a  certain privacy guarantee. 
We briefly describe the algorithm of \cite{Griggs2016} here and call it the {\em classical} algorithm. It is a distributed, iterative, and stochastic algorithm. 
In the algorithm, several clients are considered in a network, wherein each client has a cost function that depends on its time-averaged participation. Moreover, a central server is considered that keeps track of the total number of participating clients at a time step. Based on the total number of participating clients and the desired number of clients (called the {\em capacity constraint}), the central server calculates a {\em price signal} and broadcasts it in the network at each time step. After receiving the price signal, a client responds probabilistically to whether it will participate or not at the next time step. This process repeats over time. Following this process, the overall cost to the network is minimized over long-term average participation.

In the classical algorithm \cite{Griggs2016}, the clients send their true participation states to the central server, whether they would participate in completing the task at a time step or not. Let us consider scenarios where the central server or clients in the network work as an adversary. With access to the public price signals, the adversary may infer the client's participation states, average participation values, cost functions, or the derivatives of the cost functions. Therefore, we need a privacy mechanism to protect clients' privacy in the network. One of the techniques to obtain differential privacy is the randomized responses proposed by Warner \cite{Warner1965}. Randomized response is a survey technique to collect sensitive personal information of a respondent. In this technique, the respondents randomize their responses before sending them to the surveyor. The randomized response is widely used to develop differentially private algorithms, for example, as in \cite{Erlingsson2014, Ding2017, Duchi2013, Kairouz2016, Wang2016}.

Our proposed algorithm uses randomized responses wherein the clients randomize their participation states to protect their states, derivatives of the cost functions, and their cost functions. Following the algorithm, the solution to the formulated optimization problem (see \eqref{opt-prob}) provides near-optimal values with a certain privacy guarantee to clients in the network in a differential privacy sense. 
The algorithm is an extension of the classical algorithm by Griggs et al. \cite{Griggs2016} with a differential privacy guarantee using randomized response. Briefly, we consider several clients in a network; each client has a cost function, which depends on a client's time-averaged participation. Additionally, we consider a central server that keeps track of the total number of participating clients. Based on the total number of participating clients and the desired number of participants (the {\em capacity constraint}), the central server calculates a {\em price signal} and broadcasts it in the network at each time step. After receiving the price signal, a client responds probabilistically to whether it will participate or not in the next time step. If the probabilistic response is not to participate, the client further randomizes its response with a coin flip and decides whether to participate or not. This process repeats over time. By doing so, the social cost of the network over long-term average participation is close to the optimal social cost, and the central server or any adversary client will not know with certainty whether the participation states were actual states or not. Because of the additional anonymity, a certain amount of privacy is guaranteed to the clients in the network. 

{\bf Contributions:} The main contribution to this paper is the proposed differentially private algorithm for a federated setting wherein clients collaborate with a central server to complete a task without inter-agent communication. The clients decide to participate in performing the task at a time step based on their preferences and on-device computation. The algorithm is a novel modification of the classical algorithm of \cite{Griggs2016}. The algorithm solves the optimization Problem \ref{opt-prob} with close to the optimal value and provides a certain amount of privacy guarantee to the participating clients. The algorithm protects the privacy of the  client's participation states. Additionally, to check the algorithm's efficacy, we present the simulation results and compare the results with the classical algorithm \cite{Griggs2016}. 

\section{Background and Prior work}
Generally speaking, in distributed systems, the privacy of a client's states,  cost functions, sub-gradients of its cost functions, or constraints should be preserved  \cite{Ding2018, Zhang2019, Dobbe2020}. 

Han et al. \cite{Han2017} proposed a differentially private distributed algorithm for allocating divisible resources. They consider that the cost functions are convex and are continuously differentiable. To do so, they add noise to the constraints of the optimization problem. Huang et al. in  \cite{Huang2015} developed differentially private distributed algorithms using convex functions. They add noise to the cost functions of clients. A distributed differentially private algorithm was proposed by Olivier et al. \cite{Olivier2020} to optimally allocate resources. Interested readers can refer to  \cite{LeNy2014, Katewa2018, Farokhi2020, LeNy2020-book}, and \cite{Ding2018} for differentially private algorithms for dynamical systems.
Fioretto and co-authors in \cite{Fioretto2020} developed a differentially private mechanism based on Stackelberg games. Duchi et al. developed a local differential privacy mechanism in \cite{Duchi2013}. Furthermore, a local differentially private solution was developed in \cite{Dobbe2020} for solving distributed convex optimization problems. Chen and co-authors proposed two differentially private models, a local and a shuffle model, in \cite{Chen2021}. We proposed a local differential privacy algorithm for divisible resource allocation that does not require inter-agent communication in \cite{AlamIEEEAccess2023}. The current work can easily be extended to cases where clients' cost functions are multi-variate. Recently, we proposed a multi-indivisible resource allocation solution for a federated multi-agent system in \cite{Alam_CDC2023}; the proof of convergence is motivated by the multi-time scale stochastic approximation techniques, public signals there depend on the decreasing step sizes. Interested readers can also refer to \cite{Alam2022_phd} and \cite{Syed2019_unit} for the convergence proof with constant step sizes in the public signals. 
Moreover, the client selection strategies in a federated setting are studied in \cite{ZhangL2022}. Optimal control of a population of prosumers in a smart energy community is proposed in \cite{Alam_energyISC22022}. Finally, a recent survey on client selection strategies in federated learning can be found at \cite{Meth2022}.

\section{Preliminaries and problem formulation}
Let us now consider $N$ clients collaborating with a central server to complete a task. We use index $i$ for clients. Let $k \in \mathbb{N}$ denote the time step. For $k \in \mathbb{N}$ and $i=1,2,\ldots,N$, let $X_{i}(k) \in \{0,1\}$ be a random variable, it denotes whether client $i$ is participating at time step $k$ or not. Moreover, let $x_{i}(k) \in [0, 1]$ denote the average number of times a client participated up to time step $k$. For $i = 1, 2, \ldots, N$, we define $x_{i}(k)$ as follows,
\begin{align}  
{x}_{i}(k) \triangleq \frac{1}{k+1} \sum_{\ell=0}^k X_{i}(\ell).
\end{align}

Let each client have a cost function that depends on the average number of the client's participation. Furthermore, let the capacity constraint (total number of desired participation) be $\mathcal{C}$. We formulate the optimization problem in the following subsection.
\subsection{Optimization problem formulation}
Let $f_i: [0,1] \to \mathbb R_+$ be the cost function of client $i$, which associates a cost to the client. We assume that $f_i$ is twice continuously differentiable, strictly convex, and increasing for all $i$. We also assume that the clients do not share their cost functions or participation history with other clients; however, they share their participation states with the central server (whether they participate at a time step). The central server keeps track of the total number of participating clients at each time step. We formulate the following distributed/federated optimization problem.
	\begin{align} \label{opt-prob}
	\begin{split} 
	\min_{x_{1}, \ldots, x_{N}} \quad &\sum_{i=1}^{N} f_i(x_{i}),
	\\ \mbox{subject to } \quad &\sum_{i=1}^{N} x_{i}  =  \mathcal{C}, 
	\\  &x_{i} \in [0, 1], \quad i=1,\ldots,N.
	\end{split}
	\end{align}

Griggs et al. \cite{Griggs2016} proposed a distributed algorithm to solve this optimization problem with no inter-client communication; we call it  the {\em classical algorithm}. As the constraint sets of the optimization problem are compact and each client has a strictly convex cost function, a unique optimal solution exists.
Let $\mathbf{x}^* = ({x}_{1}^{*}, \ldots, {x}_{n}^*) \in (0,1]^{N}$ denote the unique solution to Problem \ref{opt-prob}. Thus,
\begin{align*}
\lim_{k\to \infty} {x}_{i}(k) = {x}_{i}^{*}, \text{ for } i=1,2,\ldots,N.
\end{align*}

We propose a local differentially private, iterative algorithm that determines whether a client is participating $X_{i}(k) \in \{0,1\}$ at a time step to complete a task or not in a federated setting. It achieves close to social minimum cost over long-term average participation with a certain privacy guarantee to clients in the network. The algorithm does not require client-to-client communication. Thus, for the solution to Problem \ref{opt-prob}, our goal is to achieve
\begin{align*}
\lim_{k\to \infty} {x}_{i}(k) \approx {x}_{i}^{*}, \text{ for } i=1,2,\ldots,N,
\end{align*}
with a certain privacy guarantee to clients in the federated network.
The proposed local differentially private algorithm is a novel modification of the distributed classical algorithm by Griggs et al. \cite{Griggs2016} based on randomized response.
For the exposition, we briefly describe the classical algorithm in Subsection \ref{single-res1}. 

\subsection{The classical algorithm} \label{single-res1} 

In the classical algorithm \cite{Griggs2016}, the idea was to choose the probability for the random variable $X_i(k)$ to ensure convergence to the socially optimum value and to adjust overall resource utilization to its capacity $\mathcal{C}$ by applying a public signal $\Theta(k)$ to the probability. When a client joins the
network at time step $k \in \mathbb{N}$, it receives the public signal
$\Theta(k)$ from the central server. At each time step $k$, the central server updates $\Theta(k)$ using a gain parameter $\tau$, past utilization of the resource, and the resource 
capacity, as in \eqref{Theta_bs1}; after updating it, the central server broadcasts the new value to all clients in the network, 
\begin{align} \label{Theta_bs1}
\begin{split}
\Theta(k+1) \triangleq \Theta(k) -  \tau  \Big (\sum_{i=1}^N X_i(k) - \mathcal{C} \Big),
\end{split}
\end{align}
\begin{align*} 
\text{where } \tau \in \Big( 0,  \Big( \max_{\mathbf{x} \in
	[0,1]^N} \sum_{i=1}^{N} \frac{ x_i}{f_i'({x}_i)}  \Big)^{-1} \Big).
\end{align*}
After receiving this signal, a client responds in a random way
based on its average allocations, the gradient of the cost function, and the public signal. The probability density function
$\sigma_i(\cdot)$ uses the average allocation of the resource to client $i$ and the derivative $f_i'$ of the cost function $f_i$, is given by,
\begin{align} \label{sigma_bs}
\sigma_i(\Theta(k),x_i(k)) \triangleq \Theta(k)
\frac{{x}_i(k)}{
	{f_i'({x}_i(k))}}, \text{ for } i=1,2,\ldots,N.	
\end{align} 
The public signal $\Theta(k)$ is chosen so that $0 < \sigma_i(\cdot) \leq 1$.
Using the probabilistic response, client $i$ updates whether it wants to participate or not in the next time step.
The process repeats over time to obtain the optimal value over long-term average clients' participation, and the network achieves minimum social cost. The algorithm of the central server is presented in Algorithm \ref{algo_CU}, and the algorithm of client $i$ is presented in Algorithm \ref{algo-QoSPA-client}.

\begin{algorithm} \SetAlgoLined Input:
	$C$, $\tau$, $ X_{i}(k)$, for $k \in \mathbb{N}$ and $i=1,2,\ldots,N$.
	
	Output:
	$\Theta(k+1)$, for $k \in \mathbb{N}$.
	
	Initialization: $\Theta(0) \in \mathbb{R}_+$,
	
	\ForEach{$k \in \mathbb{N}$}{
		
		calculate 
		$\Theta(k+1)$ as in \eqref{Theta_bs1} 
		and broadcast it in the network;	
	} 
	\caption{The classical algorithm of the central server.}
	\label{algo_CU}
\end{algorithm}

\begin{algorithm}  \SetAlgoLined Input:
	$\Theta(k)$, for $k \in \mathbb{N}$.
	
	Output: $X_{i}(k+1)$, for $k \in \mathbb{N}$.
	
	Initialization: $X_{i}(0) \leftarrow 1$ and
	${x}_{i}(0) \leftarrow X_{i}(0)$.
	
	\ForEach{$k \in \mathbb{N}$}{
		
		generate Bernoulli independent random variable $b_{i}(k)$ with the parameter  $\sigma_{i}(\Theta(k),x_i(k))$ (see  \eqref{sigma_bs});
		\begin{align} \label{eq:alloc-update}
		b_{i}(k) =
		\begin{cases} 
		1 \quad \text{with probability }  \sigma_i(\Theta(k),x_i(k));\\ 
		0 \quad \text{with probability }  1-\sigma_i(\Theta(k),x_i(k)).
		\end{cases}
		\end{align}
		
		\eIf{ $b_{i}(k) = 1$}{
			$X_{i}(k+1) \leftarrow 1$; }
		{$X_{i}(k+1) \leftarrow 0$; }

	}
	\caption{The classical algorithm of client $i$.}
	\label{algo-QoSPA-client}
\end{algorithm}

Notice that the public signal $\Theta(k)$ is a public signal broadcast by the central server at each time step. Furthermore, in the classical algorithm, the clients share their true value of participation at a time step to the central server. We consider that the central server is not trustful; it may work as an adversary, and clients' sensitive information may be leaked. We can also consider another scenario where fewer clients are in the network. A client's privacy may also be compromised if other clients in the network work as adversaries. The adversary client may gain access to the true participation values of a client. As the adversary is in the network and has access to the public signal, with this information, the adversary may infer the derivatives of the cost functions or the cost functions of a client. Therefore, we need a privacy mechanism to protect clients' privacy in the network.
Using randomized response, we propose a distributed, iterative, local differentially private algorithm that solves the optimization Problem \ref{opt-prob} and provides certain privacy guarantees to clients in the network. The proposed algorithm extends the classical Algorithm \ref{algo-QoSPA-client}. We present basic definitions and results of the differential privacy mechanisms in Section \ref{sec:DP}.  %

\subsection{Differential privacy} \label{sec:DP}
Let $S$ be a probabilistic sample space. Let $\mathcal{D}$ be a set of all possible datasets. We define a mechanism $M$ as the map $M:S \times \mathcal{D} \to \mathbb{R}$; moreover, for $D \in \mathcal{D}$, $M(D)$ represents a random variable. Let $q$ denote a query; we define it as a map $q:\mathcal{D} \to \mathbb{R}$. 
\begin{definition}[Distance between datasets \cite{Dwork2006}]
	Let $D_1, D_2 \in \mathcal{D}$ be datasets. We define the distance between the datasets  $D_1$ and $D_2$ as the smallest sample change required to change one dataset into another.
\end{definition}
\begin{definition} [Neighboring or adjacency datasets \cite{Dwork2006}]
	Let $D_1 \in \mathcal{D}$ and $D_2 \in \mathcal{D}$ be datasets. If the distance between datasets $D_1$ and $D_2$ is one, they are called neighboring or adjacency datasets, denoted by $D_1 \sim D_2$.   
\end{definition}
%
We state the following definition for $p$-norm sensitivity between datasets.
\begin{definition} ($p$-norm sensitivity \cite{Dwork2006}) Let $D_1 \in \mathcal{D}$ and $D_2 \in \mathcal{D}$ be datasets. Let query $q$ be the map $q : \mathcal{D} \to \mathbb{R}$, and let $p$-norm sensitivity be $\Delta q$; we define it as
	\begin{equation} \label{eq:sensitivity}
	\Delta q \triangleq \max_{D_1, D_2 \in \mathcal{D}} \norm{q(D_1) - q(D_2)}_p,
	\end{equation}
	for all neighboring datasets $D_1 \sim D_2$.
\end{definition}

We define $\epsilon$-differential privacy as follows.
\begin{definition}($\epsilon$-differential privacy \cite{Dwork2006})\label{def:DP} Let $S$ be the sample space and $\mathcal{D}$ be the set of datasets. Furthermore, let $D_1 \in \mathcal{D}$ and $D_2 \in \mathcal{D}$ be datasets, and let $M:S \times \mathcal{D} \to \mathbb{R}$ be a privacy mechanism and let $q : \mathcal{D} \to \mathbb{R}$ be query on $\mathcal{D}$. Then for $\epsilon \in \mathbb{R}$, and for all neighboring datasets $D_1 \sim D_2$ and for all $S \subseteq \mathbb{R}$, if the following holds
	\begin{align*}
	\mathbb{P} \left(M(D_1) \in S \right) \leq \exp{(\epsilon)} \cdot \mathbb{P} \left(M(D_2) \in S \right),
	\end{align*}
	then $M$ is called an $\epsilon$-differential privacy mechanism.
\end{definition}
Note that $\epsilon$ is called the {\em privacy budget}.
The smaller value of $\epsilon$ implies that $\mathbb{P}\left(M(D_1) \in S \right)$ and $\mathbb{P} \left(M(D_2) \in S \right)$ are close to each other, and higher privacy is protected.
Also, note that there is a trade-off between an algorithm's privacy and accuracy---the smaller $\epsilon$ provides higher privacy but lesser accuracy. In contrast, the larger $\epsilon$ provides lesser privacy but higher accuracy. 
%
%


Notice that Definition \ref{def:DP} is of the centralized systems where a  {\em central server} perturbs the public signal and sends it in the network; in this case, we assume that the central server is trustful. However, for the cases where the central server is not trustful, the {\em local differential privacy mechanism} is proposed \cite{Duchi2013}, \cite{Kasiviswanathan2008}, wherein each client runs its algorithm and perturbs its outputs before sending them to the central server. Thus, the algorithm provides a certain privacy guarantee to each client in the network. 

Let us now consider $N$ clients in a network and let the private parameters of client $i$ be stored in the dataset $D_{i}$, for $i=1,2,\ldots,N$. 
We present the definition of the local $\epsilon_i$-differential privacy mechanism $M_i$ of client $i$ as follows.
\begin{definition}[Local $\epsilon_i$-differential privacy \cite{Dobbe2020}] \label{def:Local-DP} 
	Let $S$ be the sample space and $\mathcal{D}_i$ be the set of datasets of client $i$. Furthermore, let $X_i, X_i' \in \mathcal{D}_i$ be input values, let $M_i:S \times \mathcal{D}_i \to \mathbb{R}$ be a privacy mechanism, and $q_i: \mathcal{D}_i \to \mathbb{R}$ be the query on $\mathcal{D}_i$. Then for $\epsilon_i \in \mathbb{R}$, and for all input values $X_i, X_i' \in \mathcal{D}_i$ and for all output values $\eta \in \mathbb{R}$, if the following holds
	\begin{align}
	\mathbb{P} \left(M_i(X_i) =\eta \right) \leq \exp{(\epsilon_i)} \cdot \mathbb{P} \left(M_i(X_i') =\eta \right),
	\end{align}
	then $M_i$ is called an $\epsilon_i$-local differential privacy mechanism.
\end{definition}
Here, $\epsilon_i$ is the privacy budget of client $i$. We also refer to $\epsilon_i$ as the privacy error.

\section{Local differentially private algorithm} \label{DP-local}
In this section, we propose a local differentially private algorithm for optimal client selection in a federated setting and solve optimization Problem \ref{opt-prob} with a certain privacy guarantee to each client in the network. Moreover, using the results on the sequential combination of local differentially private algorithms---wherein each client randomizes its participation intention before sending it to the central server, we show that the network also provides a certain differential privacy guarantee.

For $N$ clients in the network, each client runs its privacy mechanism to decide whether they want to participate or not in completing a task. Clients do not communicate with each other; however, they share their participation intention with the central server. The central server sends public signals to clients in the network and keeps track of the total participating clients at a time step. We consider that the central server is not trustful. Thus, to preserve privacy, a client randomizes its participation intention before sharing it with the central server. 

%
%

For client $i=1,2,\ldots,N$, recall that $f_i$ is the cost function and $f_i'$ represents the derivative of the cost function $f_i$. Also, at time step $k \in \mathbb{N}$, whether client $i$ wants to participate or not at time step $k$ is denoted by $X_i(k) \in \{0,1\}$ and time-averaged participation until time step $k$ is denoted by $x_i(k) \in [0,1]$. When client $i$ wants to participate then it updates $X_i(k) =1$ otherwise $X_i(k) =0$. Let $D_i$ be the dataset of the private information of client $i$; specifically, we define $D_i \triangleq \{f_i, f_i', x_i, X_{i} \}$. For $i=1,2,\ldots,N$, let $\mathcal{D}_i$ be the set of all possible datasets $D_i$. We define client $i$'s privacy mechanism as the map $M_i: S \times \mathcal{D}_i \to \{0,1\}$.
Let the map $q_i:D_i \to \{0,1\}$ be the query on the dataset $D_i$. The answer to the query is the client's participation intention $X_{i}(k)$ at time step $k$. The privacy mechanism of the client $i$ answering the query $q_i(D_i)$ is denoted by $M_i(D_i)$. 
Let $\bm{\beta} = (\beta_1, \beta_2, \ldots, \beta_N) \in \mathbb{R}_+^N$ be the privacy parameters known to the central server, and so it knows the maximum total additional clients' participation states at a time step. 
When a client joins the network at time step $k$, it receives a few parameters such as $\Theta(k)$ and $\beta_i \in \mathbb{R}_+$ from the central server. Note that the clients do not know the desired number of participating clients ({\em capacity constraint}) to complete the task, the total number of participating clients at a time step, or the total number of clients in the network.
However, the central server knows the capacity constraint and keeps track of the total number of participating clients at a time step. It updates and broadcasts the public signal $\Theta(k)$ (as in \eqref{Theta_bs1}) in the network. After receiving this signal, a client calculates the probability $\sigma_{i}(\Theta(k),x_i(k))$ (as in \eqref{sigma_bs}) to find Bernoulli's outcome with parameter $\sigma_{i}(\Theta(k),x_i(k))$ as follows 
\begin{align*}
b_{i}(k) = \begin{cases}
1 & \text{ with probability } \ \sigma_{i}(\Theta(k),x_i(k)),\\
0 & \text{ with probability } \ 1 -\sigma_{i}(\Theta(k),x_i(k)).
\end{cases}
\end{align*}
If the outcome $b_{i}(k)= 1$ then client $i$ shows its intention to participate and updates the participation state $X_i(k+1)=1$; otherwise, the client calculates the probability $p_i$ as in \eqref{eq:pi} based on the privacy parameter $\beta_i$ and finds Bernoulli's outcome with parameter $b_{i}'(k)$ as in \eqref{eq:pr-beta}.   
\begin{align} \label{eq:pi}
p_i \triangleq \frac{\beta_i}{2} \exp{(- \beta_i)}.
\end{align}
And,
\begin{align}   \label{eq:pr-beta}
b_{i}'(k) = \begin{cases}
1 & \text{ with probability } p_i, \\
0 & \text{ with probability } 1 - p_i.
\end{cases}
\end{align}
After obtaining the outcome $b_{i}'(k)$, client $i$ updates its participation state as follows:
\begin{align}   \label{eq:xi}
X_{i}(k+1) = \begin{cases}
1 & \text{ if } b_{i}'(k)=1, \\
0 & \text{ if } b_{i}'(k)=0.
\end{cases}
\end{align}

Every client in the network runs its algorithm to achieve close to social-optimum cost on long-term averages that is $\lim_{k\to \infty} \sum_{i=1}^N f_i(x_i(k)) \approx \sum_{i=1}^N f_i(x_i^*)$ with privacy guarantee $\epsilon_i$, for client $i$. We present client $i$'s local differential privacy result in Lemma \ref{lem:local_DP}, and the client $i$'s algorithm is presented in Algorithm \ref{algo_LDP-client}.

\begin{algorithm}  \SetAlgoLined Input:
	$\beta_i$, $\Theta(k)$, for $k \in \mathbb{N}$.
	
	Output: $X_{i}(k+1)$, for $k \in \mathbb{N}$.
	
	Initialization: $X_{i}(0) \leftarrow 1$ and
	${x}_{i}(0) \leftarrow X_{i}(0)$.
	
	\ForEach{$k \in \mathbb{N} $}{
		
		$\sigma_{i}(\Theta(k),x_i(k)) \leftarrow \Theta(k)
		\frac{{x}_{i}(k)}{f_i'(x_i(k))}$; 
		
		generate Bernoulli independent random variable
		$b_{i}(k)$ with the parameter $\sigma_{i}(\Theta(k),x_i(k))$ as in \eqref{sigma_bs};
		
		\uIf{ $b_{i}(k) = 1$}{
			$X_{i}(k+1) \leftarrow 1$;
			
		}
		\Else{generate Bernoulli independent random variable
			$b_{i}'(k)$ with the parameter $p_i$ as in \eqref{eq:pi}; 
			
			\uIf{ $b_{i}'(k) = 1$}{
				$X_{i}(k+1) \leftarrow 1$;
				
			}
			\Else{
				$X_{i}(k+1) \leftarrow 0$;}}

	}
	\caption{Differentially private Algorithm of client $i$.}
	\label{algo_LDP-client}
\end{algorithm}

\begin{figure*}
	\centering
		\includegraphics[width=0.33\textwidth]{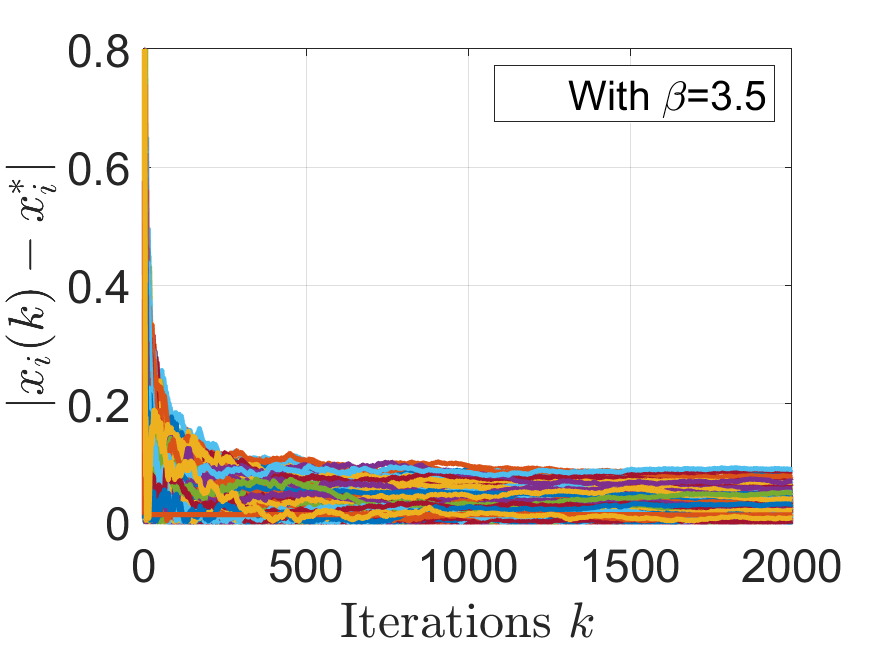}
  \subfloat[]{%
		\includegraphics[width=0.33\textwidth]{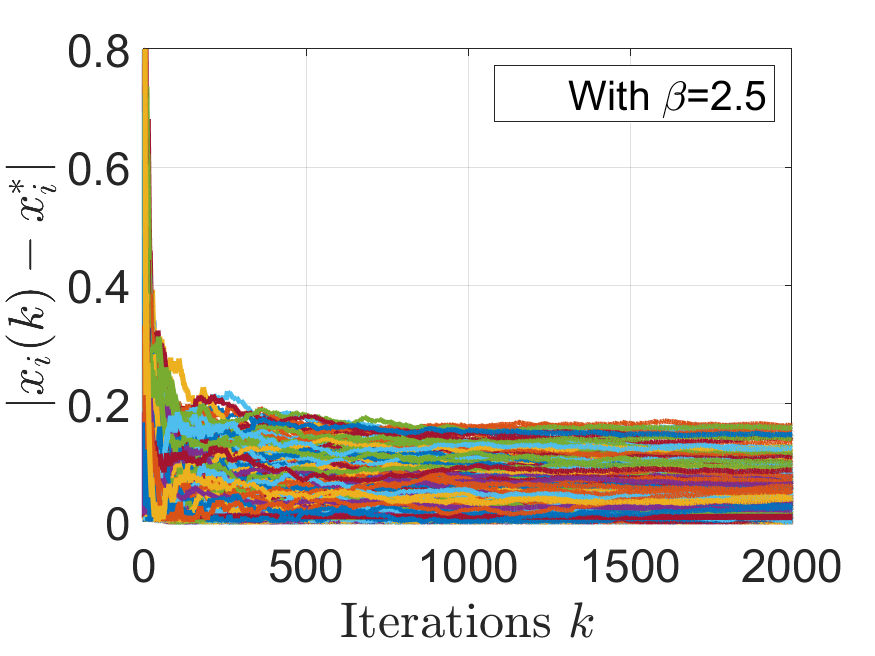}}
	\hfill
	\subfloat[]{%
		\includegraphics[width=0.33\textwidth]{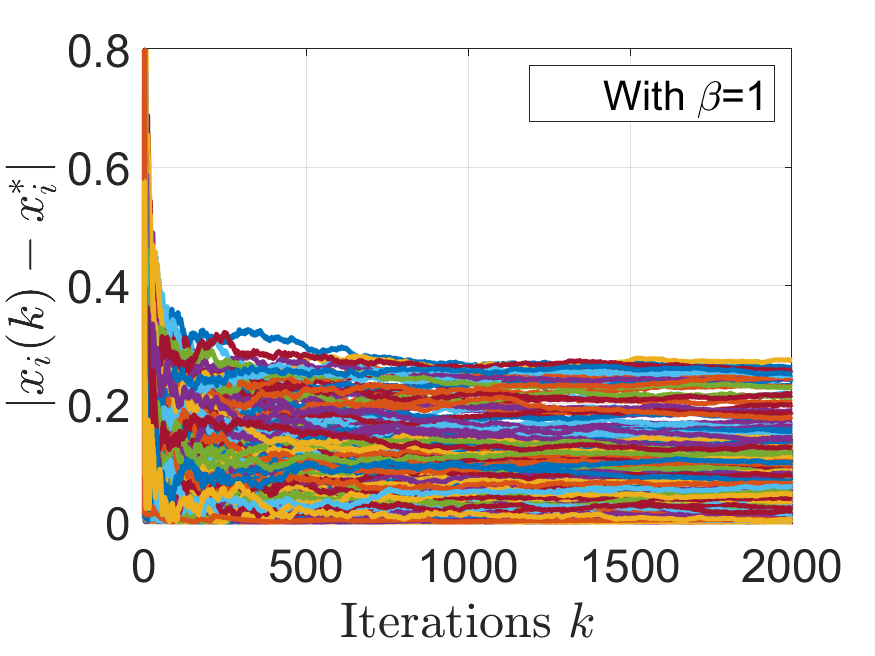}}
	\hfill
	\caption{The evolution of $|x_i(k) - x_i^*|$ for all $1200$ clients, where $x_i(k)$ is the average number of participation of clients and $x_i^*$ is the optimal value by the classical approach \cite{Griggs2016}: (a) with $\beta=3.5$, (b)  with $\beta=2.5$, and (c) with $\beta=1$.}
	\label{fig1-single} 
\end{figure*}
\begin{figure}
	\centering
	\subfloat[]{%
		\includegraphics[width=0.46\textwidth]{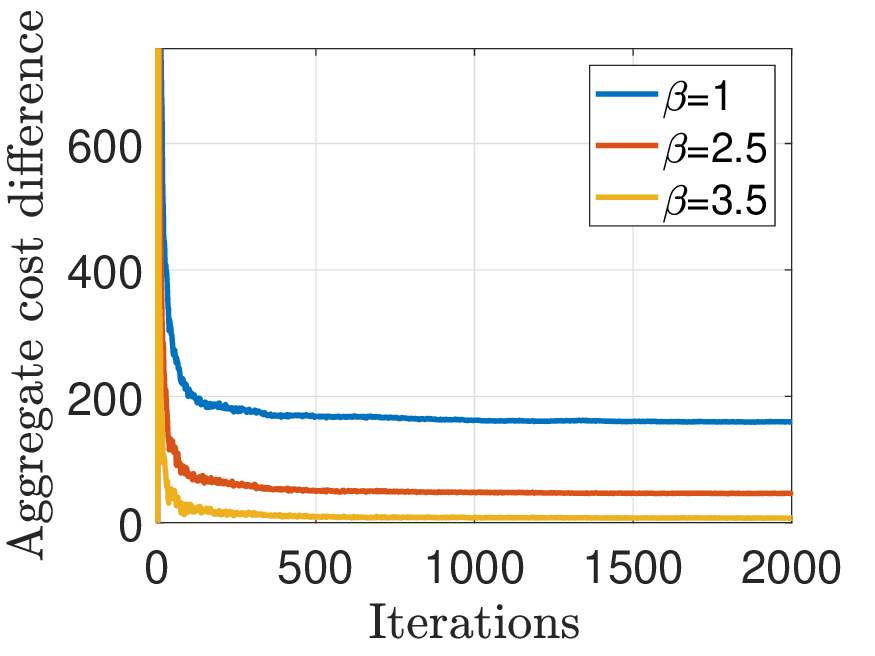}}
  \hfill
  \caption{The evolution of absolute aggregate cost difference $\mid \sum_{i=1}^N f_i(x_i(k)) - \sum_{i=1}^N f_i(x_i^*) \mid$ with privacy algorithm and the classical algorithm.}
	\label{fig3-single} 
\end{figure}
%
\begin{figure}
	\centering
		\includegraphics[width=0.46\textwidth]{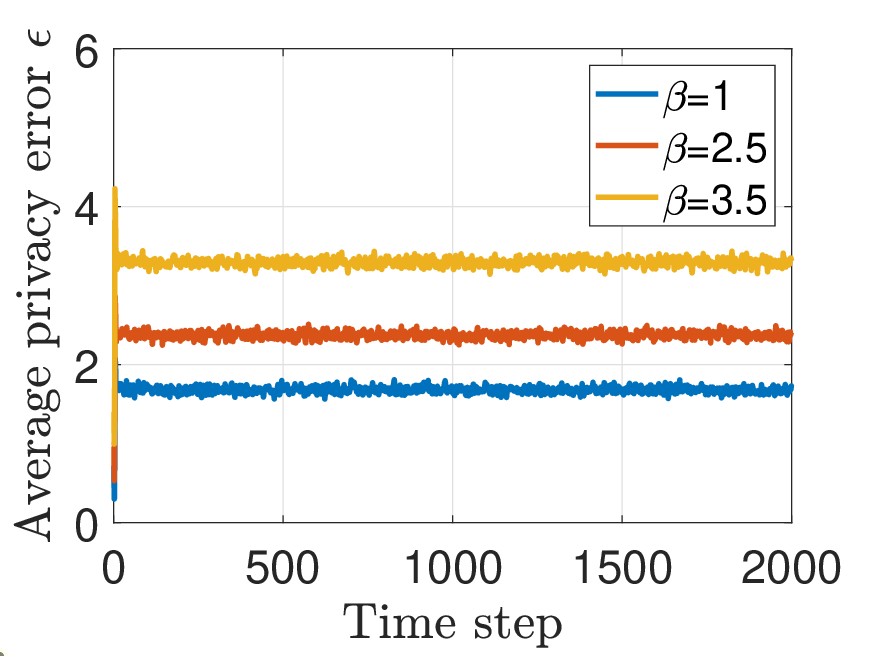}
	\caption{ The evolution of average privacy error $\epsilon(k) = \frac{1}{N} \sum_{i=1}^{N} \epsilon_i(k)$.}
	\label{fig2-single} 
\end{figure}
\begin{figure*} 
	\centering
  \subfloat[]{%
		\includegraphics[width=0.33\textwidth]{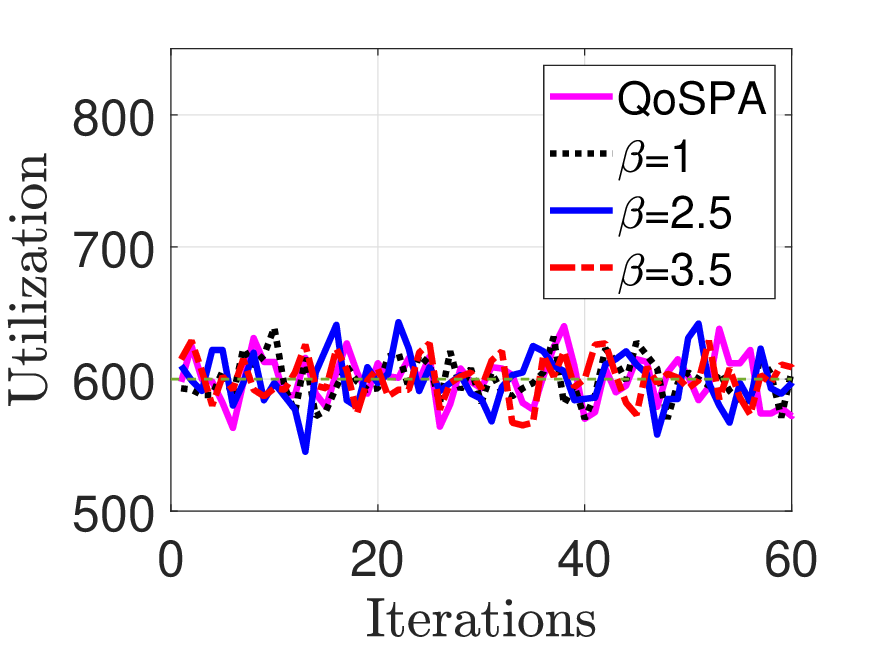}}
 \hfill
	\subfloat[]{%
		\includegraphics[width=0.33\textwidth]{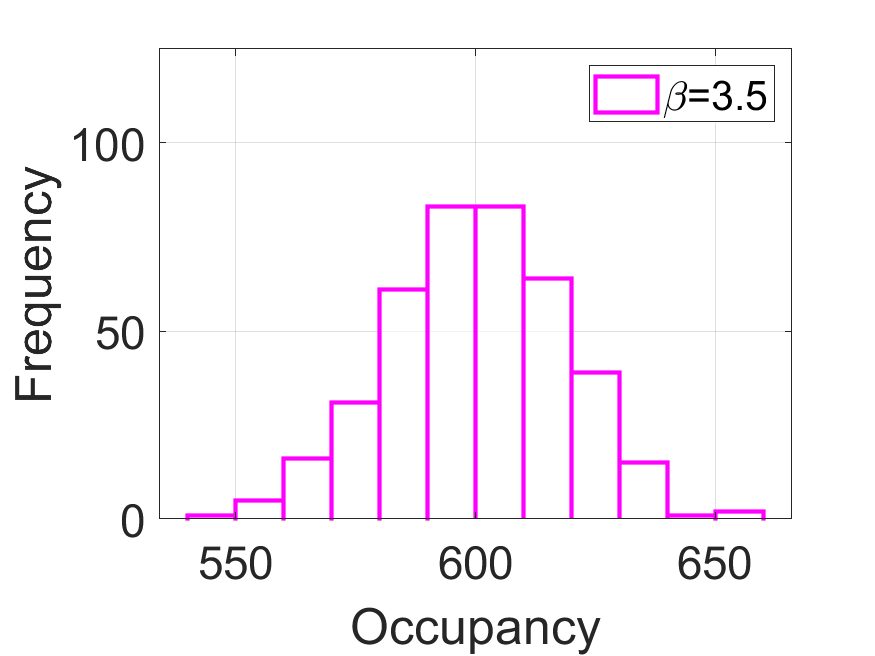}}
	\hfill
	\subfloat[]{%
		\includegraphics[width=0.33\textwidth]{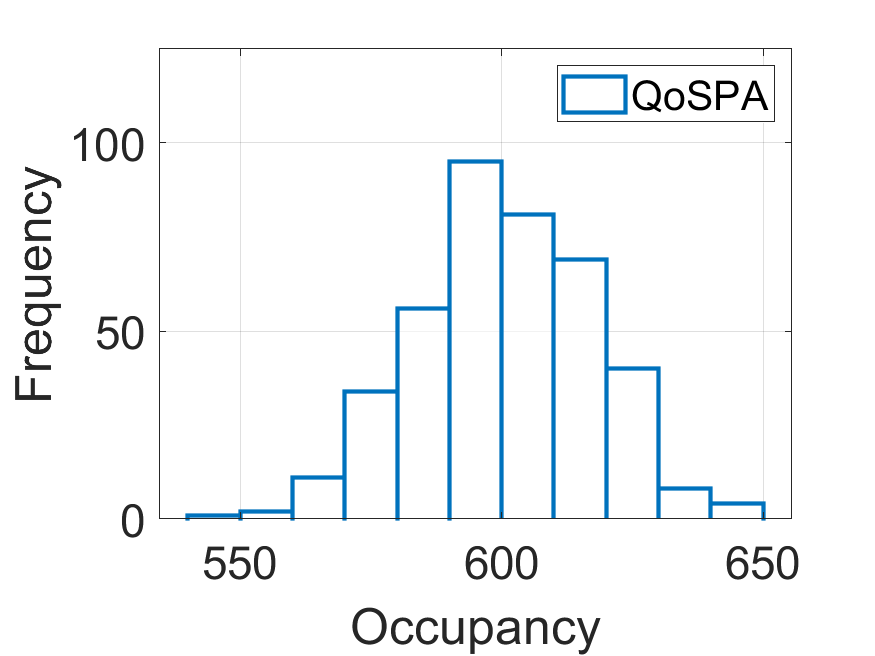}}
	\caption{ (a) The aggregate of clients' participation $\sum_{i=1}^N X_i(k)$ for different values of $\beta$ and the classical approach (QoSPA) for the chosen $60$ time steps, (b) clients' participation for $\beta=3.5$, and (c) clients' participation for the classical approach (QoSPA). Histograms $(b)$ and $(c)$ are plotted for the last $400$ time steps.}
	\label{fig5-single} 
\end{figure*}
\begin{figure*}
	\centering
	\subfloat[]{%
		\includegraphics[width=0.48\textwidth]{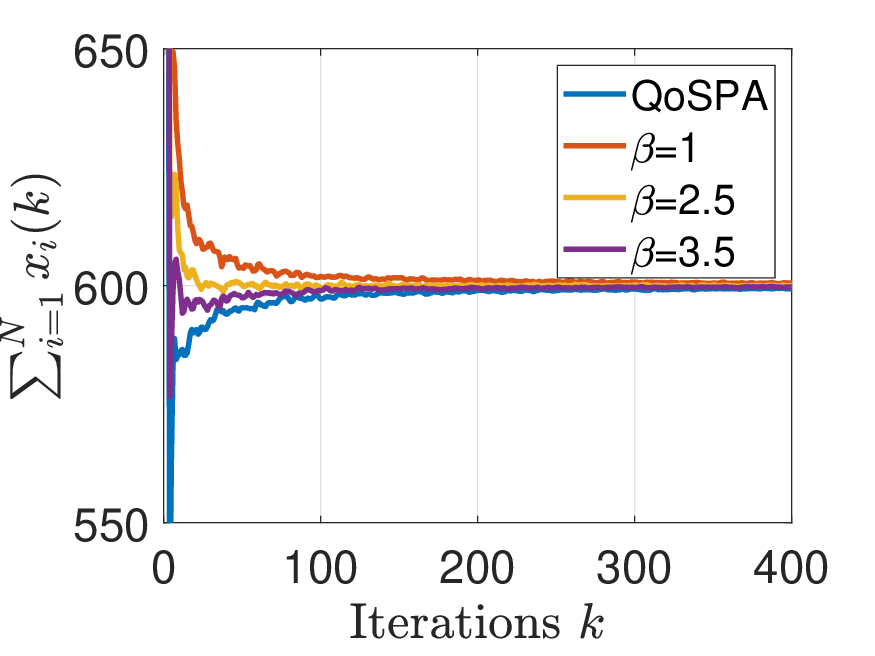}}
	\hfill
	\subfloat[]{%
		\includegraphics[width=0.48\textwidth]{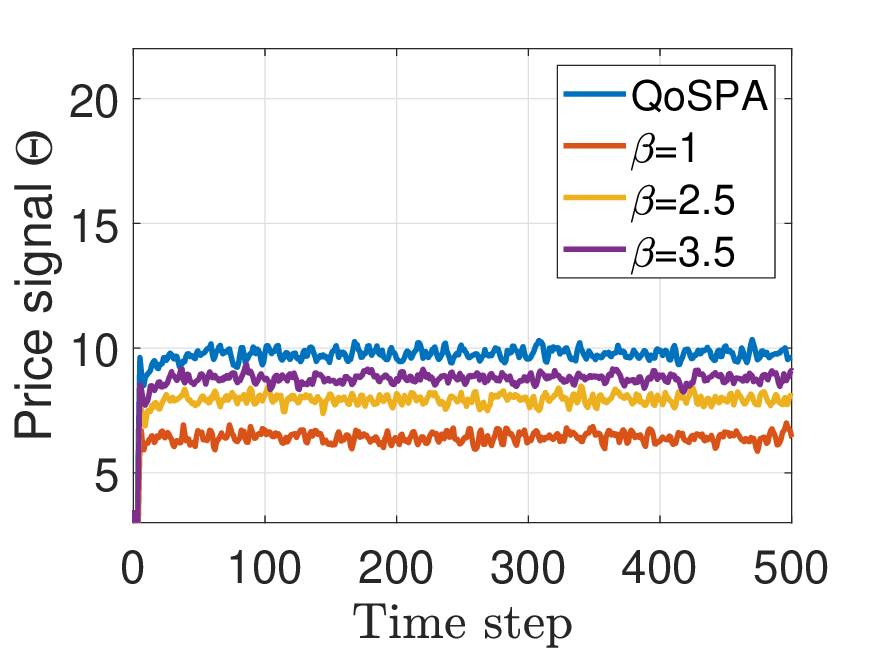}}
	\hfill
	\caption{(a) The evolution of aggregate average participation $\sum_{i=1}^N x_i(k)$, and (b) the evolution of the public signal $\Theta(k)$ with $\tau=0.0075$, for different values of privacy parameter $\beta$. QoSPA denotes the classical approach \cite{Griggs2016}.}
	\label{fig6-single} 
\end{figure*}

\begin{lemma} \label{lem:local_DP}
Let $\beta_i$ denote client $i$'s privacy parameter, $\epsilon_i(k)$ denote the privacy budget. Let the public signal be denoted by $\Theta(k)$. For $k \in \mathbb{N}$, Algorithm \ref{algo_LDP-client} is $\epsilon_i(k) \triangleq \ln \left( \frac{\Theta(k)
		{x}_{i}(k)\left( 2\exp{(\beta_i)} - \beta_i \right) +  \beta_i f_i'(x_i(k))}{\left(f_i'(x_i(k))-\Theta(k)x_i(k)
		\right) \beta_i} \right)$-local differentially private.
\end{lemma}

\begin{proof}
Following the steps similar to \cite[Chapter 3]{Dwork2014}, we obtain
	\begin{align*}
	&\left(\frac{\mathbb{P}(X_{i}(k+1) = 1 \mid X_{i}(k)=1}{\mathbb{P}(X_{i}(k+1) = 1 \mid X_{i}(k)=0} \right) \\ &= \frac{\Theta(k)
		\frac{{x}_{i}(k)}{
			f_i'(x_i(k))} + \left(1-\Theta(k)
		\frac{{x}_{i}(k)}{
			f_i'(x_i(k))}\right) \times p_i}{\left(1-\Theta(k)
		\frac{{x}_{i}(k)}{
			f_i'(x_i(k))}\right)p_i}.
	\end{align*}
	Replacing the value of $p_i$ (see Equation \eqref{eq:pi}), we obtain the following result:
\begin{align*}
	&\left(\frac{\mathbb{P}(X_{i}(k+1) = 1 \mid X_{i}(k)=1}{\mathbb{P}(X_{i}(k+1) = 1 \mid X_{i}(k)=0} \right) \\ &= \frac{\Theta(k)
		\frac{{x}_{i}(k)}{
			f_i'(x_i(k))} + \left(1-\Theta(k)
		\frac{{x}_{i}(k)}{
			f_i'(x_i(k))}\right) \frac{\beta_i}{2} \exp{(- \beta_i)}}{\left(1-\Theta(k)
		\frac{{x}_{i}(k)}{
			f_i'(x_i(k))}\right)(\frac{\beta_i}{2} \exp{(- \beta_i)})} \\
   &= \frac{\Theta(k)
		\frac{{x}_{i}(k)}{
			f_i'(x_i(k))} + \left(\frac{f_i'(x_i(k))-\Theta(k)
		{x}_{i}(k)}{
			f_i'(x_i(k))}\right) \frac{\beta_i}{2\exp{( \beta_i)}} }{\left(\frac{f_i'(x_i(k))-\Theta(k)
		{x}_{i}(k)}{
			f_i'(x_i(k))}\right) \frac{\beta_i}{2\exp{(\beta_i)}} }
	\end{align*}
\begin{align*}
       &= \frac{2\exp{(\beta_i)}\Theta(k)
		{x}_{i}(k) + \left(f_i'(x_i(k))-\Theta(k)
		{x}_{i}(k)\right) \beta_i }{\left(f_i'(x_i(k))-\Theta(k)
		{x}_{i}(k)\right) \beta_i}
  \\
   &= \frac{\Theta(k)
		{x}_{i}(k)\left(2 \exp{(\beta_i)}- \beta_i\right) + \beta_i f_i'(x_i(k)) }{\left(f_i'(x_i(k))-\Theta(k)
		{x}_{i}(k)\right) \beta_i}
  \\&\triangleq \exp{(\epsilon_i(k))}.
\end{align*} 

Thus, we obtain the following privacy error (budget) of client $i$ at time step $k$:
	\begin{align*}
	\epsilon_i(k) = \ln \left( \frac{\Theta(k)
		{x}_{i}(k)\left(2 \exp{(\beta_i)}- \beta_i\right) + \beta_i f_i'(x_i(k)) }{\left(f_i'(x_i(k))-\Theta(k)
		{x}_{i}(k)\right) \beta_i} \right).
	\end{align*}
\end{proof}

%

%

We state the following result on the sequential combination of privacy budgets of a population of clients. 
\begin{theorem}[Sequential combination~\cite{Dwork2014}] \label{th:sequential}
	For clients $i=1,2,\ldots,N$ and for $\epsilon_i \in \mathbb{R}$, if $M_i$ is $\epsilon_i$-differentially private mechanism then the resulting sequential combination mechanism $M \triangleq (M_1, \ldots, M_N)$ will be $\sum_{i=1}^{N}\epsilon_i$-differentially private.
\end{theorem}

We present the following result on the average of the differential privacy of the federated network.
\begin{theorem} 
	Let there be $N$ clients in a federated network; each client runs its local differentially private Algorithm \ref{algo_LDP-client}, then the network is $ \frac{1}{N} \sum_{i=1}^N \ln \left( \frac{\Theta(k)
		{x}_{i}(k)\left(2 \exp{(\beta_i)}- \beta_i\right) + \beta_i f_i'(x_i(k)) }{\left(f_i'(x_i(k))-\Theta(k)
		{x}_{i}(k)\right) \beta_i} \right)$-differentially private on average.
\end{theorem}

\begin{proof}
	Using Lemma \ref{lem:local_DP} and the result sequential result of Theorem \ref{th:sequential}, it is straightforward to obtain the result.
\end{proof}

\section{Experimental results}
In this section, we describe the experimental setup and the results. We observe that the clients achieve near-optimal value over long-term average participation and minimize the overall cost to the network with a differential privacy guarantee. We consider $N=1200$ clients in a network that collaborate with a central server to perform a task. Let the required number of active clients to perform the task at a time step be $\mathcal{C} = 600$. 
%
For clients $i=1,2,\ldots,N$, let $a_i$ and $b_i$ be uniformly distributed random variables in $(0,40)$. We consider the following cost functions for clients: 
\begin{equation} \label{bin_func} f_{i}(x_{i})= \left\{
\begin{array}{ll}
(i) \hspace{0.1in}  &a_i(x_{i})^2, \\
(ii) \hspace{0.1in}  &\frac{1}{2} a_i(x_{i})^4, \\
(iii) \hspace{0.1in} &\frac{1}{3}a_i(x_{i})^4 + b_i(x_{i})^6,
\\
(iv) \hspace{0.1in} &b_i(x_{i})^2.
\end{array}
\right.
\end{equation}
The clients are categorized into four groups; the first group's clients have cost functions listed in \eqref{bin_func} $(i)$. Analogously, the second, the third, and the fourth group's clients have cost functions listed in \eqref{bin_func} $(ii)$, \eqref{bin_func} $(iii)$, and \eqref{bin_func} $(iv)$, respectively.

We now present the experimental results. We observe that the average number of clients' participation $x_i(k)$  with a certain privacy guarantee by our approach converges close to the optimal participation $x_i^*$ by the classical approach of \cite{Griggs2016}. Note that the value of $x_i^*$ is the average number of clients' participation by the classical approach at the last time step of the simulation. When $\beta_i$ decreases, privacy increases, but efficiency decreases, as illustrated in Figure \ref{fig1-single}. Figure \ref{fig1-single} (a) shows the absolute difference of average participation of clients $\mid x_i(k) - x_i^* \mid$ by our approach with privacy parameter $\beta_i= 3.5$ and the classical approach \cite{Griggs2016}. The absolute difference $\mid x_i(k) - x_i^* \mid$ is close to $0.1$ for all $N=1200$ clients. Notice that we use the same value of $\beta_i$ for all the clients in the network, that is $\beta_i = \beta_u = \beta = 3.5$, for $i,u \in \{1,2,\ldots,N\}$. 
However, when the privacy increases (that is, the value of $\beta_i$ decreases), the absolute difference $\mid x_i(k) - x_i^* \mid$ increases, as demonstrated in Figure \ref{fig1-single} (b) and Figure \ref{fig1-single} (c). It signifies that when $\beta_i$ decreases, the average participation value $x_i(k)$ goes farther from the optimal value $x_i^*$; hence, efficiency decreases. Moreover, Figure \ref{fig1-single} (b) shows the evolution of the absolute difference $\mid x_i(k) - x_i^* \mid$ with $\beta=2.5$ and Figure \ref{fig1-single} (c) shows the evolution of  the absolute difference $\mid x_i(k) - x_i^* \mid$ with $\beta=1.5$.    
%
%
The same conclusion can be derived from Figure \ref{fig3-single}: when $\beta_i$ is small, the cost difference is high; the solution will be farther from the optimal points, but more privacy is guaranteed. Analogously,  when $\beta_i$ increases, the cost difference decreases; thus, the solution will come closer to the optimal points, but less privacy will be guaranteed.

Figure \ref{fig2-single} illustrates the evolution of the average privacy error (budget) for different values of $\beta$. We observe that the privacy error $\epsilon_i$ is small when $\beta$ is small; nevertheless, it increases with the increase of $\beta$. Recall that privacy decreases when $\epsilon_i$ increases, but the algorithm's efficiency increases.
%
%
%
%
Figure \ref{fig5-single}(a) illustrates that the total number of participating clients $\sum_{i=1}^N X_i(k)$ for different values of $\beta$ concentrate around the capacity constraint $\mathcal{C}$. Moreover, the histogram plotted for the last $400$ time steps in Figure \ref{fig5-single}(b) shows that most occupancy is closer to the capacity constraint $\mathcal{C}=600$, as also in Figure \ref{fig5-single}(c) by the classical algorithm \cite{Griggs2016}. 
In Figure \ref{fig6-single}(a), the aggregate average participation of clients $\sum_{i=1}^N x_i(k)$ converges to the capacity constraint $\mathcal{C}$. Finally, Figure \ref{fig6-single}(b) shows the evolution of the public signal $\Theta(k)$ with different values of the privacy parameters $\beta$.

\section{Conclusion}
We developed a differentially private algorithm for client selection in federated settings. The algorithm provides near-optimal values over the long-term average participation of clients and provides a  certain differential privacy guarantee to clients in the network. 
Applying the algorithm or its variant for client selection in federated learning is an interesting problem. 

\section*{Acknowledgement}
Thanks to Jia Yuan Yu for the early discussions. This work is partially supported by Mitacs (grant number IT24468).

\bibliographystyle{IEEEtran}
\bibliography{DP_Indiv-Bibtex}


\end{document}